\documentclass[12pt,twoside]{article}
\usepackage[mathscr]{eucal}
\usepackage{amsthm,amsmath,amssymb,amscd}
\usepackage{graphicx}
\usepackage{latexsym,enumerate}
\usepackage{color}
\usepackage{tikz}
\usetikzlibrary{shapes.arrows,chains,positioning}
\usepackage[outline]{contour}
\usepackage{caption}
\usepackage{subcaption}
\newtheorem{theorem}{Theorem}
\newtheorem{lemma}{Lemma}
\newtheorem{proposition}{Proposition}
\newtheorem{corollary}{Corollary}
\newtheorem{condition}{Planarity Condition}

\theoremstyle{definition}

\newtheorem*{remark}{Remark}

\usepackage{color,hyperref}
\definecolor{dark-blue}{rgb}{0.15,0.15,0.4}
\definecolor{dark-red}{rgb}{0.4,0.15,0.15}
\definecolor{medium-red}{rgb}{0.6,0,0}
\definecolor{medium-blue}{rgb}{0,0,0.6}
\hypersetup{
    colorlinks=true,
    linktoc=all,
    citecolor=medium-red,
    filecolor=medium-blue,
    linkcolor=medium-blue,
    urlcolor=medium-blue,
}
\begin{document}
\title{Planarity conditions and four-body central configurations equations with angles as coordinates}
\author{Manuele Santoprete\thanks{ Department of Mathematics, Wilfrid Laurier
University E-mail: msantopr@wlu.ca}} 
 \maketitle

\begin{abstract}
    We discuss several conditions for four points to lie on a plane, and we use them to find new equations for four-body central configurations that use angles as variables. We use these equations to give novel proofs of some results for  four-body central configuration. We also give a clear geometrical  explanation of why Ptolemy's theorem can be used to write equations for co-circular central configurations when  mutual distances are used as coordinates. 
\end{abstract}

\renewcommand{\thefootnote}{\alph{footnote})}
\tableofcontents
%\maketitle

\section{Introduction}
Let $ P _1,P _2, P _3$, and $P _4 $ be  four points in $\mathbb{R}^3$ with position vectors $\mathbf{q} _1 , \mathbf{q} _2 , \mathbf{q} _3  $, and $ \mathbf{q} _4 $, respectively (see figure \ref{fig:tetrahedron}). Let $ r _{ ij } = \| \mathbf{q} _i - \mathbf{q} _j \| $, be the distance between the point $ P _i $ and $ P _j $, and let $ \mathbf{q} = (\mathbf{q} _1 , \mathbf{q} _2 , \mathbf{q} _3 , \mathbf{q} _4) \in \mathbb{R}  ^{ 12 }$.
The center of mass of the system is $ \mathbf{q} _{ CM } = \frac{ 1 } { m ' } \sum _{ i = 1 } ^n m _i \mathbf{q} _i $, where $ m ' = m _1 + \ldots m _n $ is the total mass. 
The Newtonian $4$-body problem concerns the motion of $4$ particles
with masses $m_i\in{\mathbb R}^+$ and positions $\mathbf{q} _i\in{\mathbb
R}^3$, where $i=1,\ldots,4$. The motion is governed by Newton's
law of motion 
\begin{equation} 
 m _i\mathbf{\ddot q} _i=   \sum _{ i \neq j } \frac{m _i m _j (\mathbf{q} _j - \mathbf{q} _i)   } { r _{ ij } ^3 }=\frac{\partial  U}{  \partial  \mathbf{q}  _i}, \quad 1\leq i\leq 4
\end{equation} 
where $U(\mathbf{q} )$ is the Newtonian potential
\begin{equation} 
U( \mathbf{q} )=\sum_{i<j}\frac{m_im_j}{r_{ij}},\quad 1\leq i\leq 4.
\end{equation} 
A {\it central configuration} (c.c.) of the four-body problem is a configuration $ \mathbf{q} \in \mathbb{R} ^{ 12} $ which satisfies the algebraic equations
%\begin{equation} \label{eqn:cc1}
%    \nabla U (\mathbf{q}) = \lambda M \mathbf{q}  
%\end{equation}
%for some value of $ \lambda $,where $ M = \operatorname{diag} (m _1 , m _1 , m _1 , \ldots, m _n , m _n , m _n) $ 
%is the diagonal mass matrix. 
\begin{equation} \label{eqn:cc1}
 \lambda\, m _i (\mathbf{q} _i -\mathbf{q} _{ CM })  = \sum _{ i \neq j } \frac{ m _i m _j (\mathbf{q} _j - \mathbf{q} _i) } { r _{ ij } ^3 }, \quad 1 \leq i \leq n .% \quad i = 1, \ldots 4.
 \end{equation} 
If we let $ I (\mathbf{q}) $ denote the moment of inertia, that is,
\[ I (\mathbf{q}) = \frac{1}{2} \sum _{ i = 1 } ^n m _i \| \mathbf{q} _i - \mathbf{q} _{ CM } \| ^2 = \frac{ 1}{2 m'} \sum _{ 1 \leq i < j \leq n } ^n  m _i m _j r _{ ij } ^2, %\quad \mbox{where} ~  m' = \sum _{ i = 1 } ^n m _i
\]
    we can write equations \eqref{eqn:cc1} as 
    \begin{equation}\label{eqn:cc3}  \nabla U (\mathbf{q}) = \lambda \nabla I (\mathbf{q}).  
    \end{equation}
   Viewing $ \lambda $  as a Lagrange multiplier, a central configuration is simply a critical point of $U $
subject to the constraint $I$ equals a constant.

A central configuration is {\it planar} if the four points $ P _1,P _2, P _3$, and $P _4 $  lie on the same plane. Equations \eqref{eqn:cc1}, and \eqref{eqn:cc3}  also describe planar central configurations provided $ \mathbf{q} _i \in \mathbb{R}^2  $ for $ i = 1, \ldots 4 $.

Other equations for the planar  central configurations of four bodies were given by Dziobek \cite{dziobek1900uber}.
Such equations are written in  term of mutual distances, and  to obtain them, one must  include a condition ensuring that the configuration is planar. This is usually done by using the variational approach of Dziobek (good references for this approach are  \cite{dziobek1900uber,schmidt2002central,perez2007convex}, and \cite{hampton2014relative} for the vortex case) and the following planarity condition
\begin{condition}
 $ P _1,P _2, P _3, P _4 \in \mathbb{R}^3$  are coplanar (in the same plane) if and only if the volume of the tetrahedron formed by  these four points is $ 0 $ .
\end{condition}
To use this condition explicitly  one typically sets the Cayley-Menger determinant 
\[e _{ CM }  = \begin{vmatrix}
            0 & 1 & 1 & 1 & 1 \\
            1 & 0 & r^2 _{ 12 } & r^2 _{ 13 } & r^2 _{ 14 }  \\
            1 & r ^2_{ 12 } & 0 & r^2 _{ 23 } & r^2 _{ 24 } \\
            1 & r ^2_{ 13 } & r^2 _{ 23 } & 0 & r^2 _{ 34 } \\
            1 & r^2 _{ 14 } & r^2 _{ 24 } & r^2 _{ 34 } & 0
        \end{vmatrix} 
\]
to zero. 
Other authors (see for example \cite{albouy2006mutual,albouy2008symmetry,moeckel2001generic}), however, derive Dziobek's equations  using another approach based on  another  planarity condition:

\begin{condition}
 The dimension of the configuration determined by the points $ P _1,P _2, P _3, P _4 \in \mathbb{R}^3$  is less or equal $ 2 $  if and only if there is a non zero vector $ A = (A _1 , A _2 , A _3 , A _4) $ such that.  
 \begin{align*}
    A _1 + A _2 + A _3 + A _4 & = 0\\
   A _1 
   \mathbf{q}  _1 + A _2 
   \mathbf{q}  _2 + A _3 
   \mathbf{q}  _3 + A _4 
   \mathbf{q}  _4 & = 0.
 \end{align*} 
 Moreover, the dimension of the configuration is $ 2 $ if and only if $A$ is unique up to a constant factor. 
\end{condition}

The main purpose of this paper is to describe in detail some lesser known planarity conditions and apply them to recover some known results. We also want to give a geometrical explanation of the constraints used by Cors and Roberts \cite{cors2012four} to obtain co-circular c.c.'s. In Section 2 we will study three more planarity conditions (i.e.,Planarity Condition 3, 4, and 5). In Section 3 using one of the lesser known planarigy conditions (i.e.,Planarity Condition 4), we obtain some new equations for four-body c.c's in terms of angles. A different set of c.c's equations in terms of angles was first obtained by  Saari  \cite{saari2005collisions}. The  equations obtained by Saari  use  Planarity Condition 3 and thus are only suited to describe concave configurations. 
In Section 4 we present some applications of the  equations obtained is Section 3 to prove some known results. In particular, we  give a different proof  of Lemma 3.2 in \cite{cors2012four}, namely we show that if two pairs of masses are equal then the corresponding co-circular configuration is an isosceles trapezoid. We also use the new equations to give a proof of  Lemma 2.5  in \cite{hampton2014relative}.
In Section 5 we give a geometrical explanation of why the equations for four-body co-circular central configuration with distances as variables can be obtained  using Ptolemy's theorem as a constraint. 
\section{More Planarity Conditions}
In this section we will explore some lesser known planarity conditions from a purely geometric point of view. In Sections, 3, 4, and 5  we will investigate their applications to the four body problem. 
Let $ P _1,P _2, P _3$, and $P _4 $ be  four points in $\mathbb{R}^3$ and let  $\mathbf{q} _1 , \mathbf{q} _2 , \mathbf{q} _3  $, and $ \mathbf{q} _4 $ be their position vectors.  Let  
\begin{align*} 
    \mathbf{a}&  = \mathbf{q}  _2 - \mathbf{q}  _1 ,~ 
    \mathbf{b}  = \mathbf{q}  _3 - \mathbf{q}  _2 , ~
     \mathbf{c}   = \mathbf{q}  _4 - \mathbf{q}  _3 , \\
     \mathbf{d}  & = \mathbf{q}  _1 - \mathbf{q}  _4 , ~ 
     \mathbf{e}  = \mathbf{q}  _3 - \mathbf{q}  _1 , ~
    \mathbf{f}  = \mathbf{q}  _4 - \mathbf{q}  _2,    
\end{align*} 
then it follows that  $ \mathbf{a} + \mathbf{b} + \mathbf{c} + \mathbf{d}= 0 $,  $ \mathbf{f} = \mathbf{b} + \mathbf{c} $, and $ \mathbf{e} = \mathbf{a} + \mathbf{b} $, see figure \ref{fig:tetrahedron}.

%be the mutual distances between the points. Let us associate to the distances $a, b,c,d,e,f $ the vectors $ \mathbf{a} , \mathbf{b} , \mathbf{c} , \mathbf{d} , \mathbf{e}, \mathbf{f} $ so arranged that $ \mathbf{a} + \mathbf{b} + \mathbf{c} + \mathbf{d}= 0 $, and $ \mathbf{f} = \mathbf{b} + \mathbf{c} $, $ \mathbf{e} = \mathbf{a} + \mathbf{b} $.
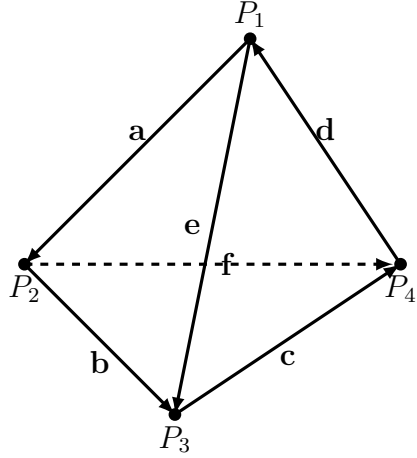
\begin{figure}[t]\begin{center}
\begin{tikzpicture}[thick]%,scale=0.8, every node/.style={scale=1}]
\clip (-2.5,-2.5) rectangle (3.5,3.6);
%\draw[thin,color=white!40!gray]  (-4.9,-4.9) grid (4.9,4.9);
%\draw[->,very thick] (-5,0) -- (5,0) node[below] {$x$};
%\draw[->,very thick] (0,-5) -- (0,5) node[above] {$y$};
%\draw[very thick,color=blue] (-2,0) -- (0,3) node[below] {};
%\draw[very thick,color=blue] (-4,3) -- (-2,0) node[below] {};
%\draw[color=blue,fill=blue] (-4,3) circle (0.1);
%\draw[color=blue,fill=blue] (0,3) circle (0.1);
%\draw[color=blue,fill=blue] (-2,0) circle (0.1);
\draw[very thick,color=black,>=latex,->] (1,3) --   (-2,0)node[midway,above] {$ \mathbf{a}$  };%; node[below] {};%a
\draw[very thick,color=black,>=latex,->] (-2,0) -- (0,-2) node[midway,below] {$\mathbf{b}$ };%b
\draw[very thick,color=black,>=latex,<-] (3,0) -- (0,-2) node[midway,below] {$\mathbf{c}$ };%c

\draw[very thick,color=black,>=latex,<-] (1,3) -- (3,0) node[midway,above] {$\mathbf{d}$ };%d
\draw[very thick,color=black,>=latex,->] (1,3) -- (0,-2) node[midway,left] {$\mathbf{e}$ };%e
\draw[very thick,dashed,color=black,>=latex,->] (-2,0) -- (2.9,0) node[midway,right] {$\mathbf{f}$ };%f
\draw[color=black,fill=black] (1,3) circle (0.07)  node[above] {$ P _1 $ };
\draw[color=black,fill=black] (-2,0) circle (0.07)  node[below] {$ P_2 $ };
\draw[color=black,fill=black] (0,-2) circle (0.07)  node[below] {$ P _3 $ };
\draw[color=black,fill=black] (3,0) circle (0.07)  node[below] {$ P _4 $ };
%\draw (1,3) arc (240:275:1cm);
%\draw[thick,red] ([shift=(225:0.8cm)]1,3) arc (225:305:0.8cm);
%\draw[thick,red] ([shift=(35:0.8cm)]0,-2) arc (35:135:0.8cm);

%\draw[very thick] (1.2,2.3) node[below] {$ \alpha $ };
%\draw[very thick] (0.5,-0.8) node[below] {$ \beta$ };

\end{tikzpicture}
\end{center}
\caption{The points $ P _1, P _2 , P _3$ , and $P _4$  form a tetrahedron in $ \mathbb{R}  ^3 $.\label{fig:tetrahedron}}
\end{figure}

\begin{lemma}\label{lemma:ef}
    Let $ \Delta = \frac{1}{2} \| \mathbf{e} \times \mathbf{f} \| $, then, with the above definitions, the following equation holds
    \[
        4\Delta  ^2 = e ^2 f ^2 - \frac{1}{4} (b ^2 + d ^2 - a ^2 - c ^2 ) ^2,
    \]
    or 
    \[
        \Delta ^2  = (s -a) (s - b) (s - c) (s - d) - \frac{1}{4} (a c + b d + ef)  (ac + b d - e f)       
    \]
\end{lemma} 
\begin{proof} 
 Clearly,
\[
    4\Delta ^2  = (\mathbf{e} \times \mathbf{f}) \cdot (\mathbf{e} \times \mathbf{f}) =  (\mathbf{e} \cdot \mathbf{e }) (\mathbf{f} \cdot \mathbf{f}) -  (\mathbf{e} \cdot \mathbf{f} ) ^2 = e ^2 f ^2 -  (\mathbf{e} \cdot \mathbf{f}) ^2.       
\]
But 
\begin{align*}
   2 (\mathbf{e} \cdot \mathbf{f}) & = 2\cdot (\mathbf{a} + \mathbf{b} ) \cdot (\mathbf{b} + \mathbf{c})  = - 2 \mathbf{b} \cdot (\mathbf{c} + \mathbf{d}) + 2 \mathbf{c} \cdot (\mathbf{a} + \mathbf{b})  = 2 \mathbf{a} \cdot \mathbf{c} - 2 \mathbf{b} \cdot \mathbf{d}\\
   &= (\mathbf{a} + \mathbf{c}) \cdot (\mathbf{a} + \mathbf{c}) - \mathbf{a} \cdot \mathbf{a} - \mathbf{c} \cdot \mathbf{c} - (\mathbf{b} + \mathbf{d}) \cdot (\mathbf{b} + \mathbf{d}) + \mathbf{b} \cdot \mathbf{b} + \mathbf{d} \cdot \mathbf{d} \\
   & = -a ^2 + b ^2 - c ^2 + d^2. 
\end{align*}  
Hence, 
\[
    4\,\Delta  ^2 = e ^2  f ^2 - \frac{1}{4} (b ^2 + d ^2 - a ^2 - c ^2 ) ^2.  
\] 
\end{proof}

Suppose  that   $\mathbf{q} _1 , \mathbf{q} _2 , \mathbf{q} _3,\mathbf{q} _4 \in \mathbb{R}  ^2 $, then we say that the configuration is {\it planar}. We say that a planar configuration is {\it degenerate} if two or more points coincide, or if more than two points lie on the same line. 
Non-degenerate planar configurations can be classified as either {\it concave} or {\it convex}  (see Figure \ref{fig:concave}, and \ref{fig:convex}).
A concave configuration has one point which is located strictly inside the convex hull of the other three, whereas a convex configuration does not have a point contained in the convex hull of the other three points. Any  convex configuration determines a convex quadrilateral (for a precise definition  of quadrilateral see for example \cite{behnke1974fundamentals}). {\it Co-circular} central configurations are those planar four-body
c.c.’s which also lie on a common circle.

In a planar convex configuration we say that the points are {\it ordered sequentially} if they are numbered consecutively while traversing the boundary of the corresponding convex  quadrilateral. 
We say that a planar convex configuration is {\it standard} (see Figure \ref{fig:convex}) if it is  a non-degenerate planar convex configuration that is ordered sequentially so that $ r _{ 12 },  r _{ 23 } , r _{ 34 } $ and $ r _{ 14 } $ are the lengths of the exterior sides of the corresponding  quadrilateral and $ r _{ 13 } $ and $ r _{ 24 } $ are the lengths of the diagonals. A {\it standard quadrilateral} is the quadrilateral determined by a standard configuration.

\begin{figure}
\begin{subfigure}{.45\textwidth}
    \centering
%\begin{center}
\begin{tikzpicture}[thick]%,scale=0.8, every node/.style={scale=1}]
\clip (-2.5,-2.5) rectangle (3.5,3.6);
%\draw[thin,color=white!40!gray]  (-4.9,-4.9) grid (4.9,4.9);
%\draw[->,very thick] (-5,0) -- (5,0) node[below] {$x$};
%\draw[->,very thick] (0,-5) -- (0,5) node[above] {$y$};
%\draw[very thick,color=blue] (-2,0) -- (0,3) node[below] {};
%\draw[very thick,color=blue] (-4,3) -- (-2,0) node[below] {};
%\draw[color=blue,fill=blue] (-4,3) circle (0.1);
%\draw[color=blue,fill=blue] (0,3) circle (0.1);
%\draw[color=blue,fill=blue] (-2,0) circle (0.1);
\draw[very thick,color=black,>=latex,->] (1,3) --   (-2,0)node[midway,above left ] {$ \mathbf{a}$  };%; node[below] {};%a
\draw[very thick,color=black,>=latex,->] (-2,0) -- (0.3,1) node[midway,above] {$\mathbf{b}$ };%b
\draw[very thick,color=black,>=latex,<-] (3,0) -- (0.3,1) node[midway,above] {$\mathbf{c}$ };%c

\draw[very thick,color=black,>=latex,<-] (1,3) -- (3,0) node[midway,above right] {$\mathbf{d}$ };%d
\draw[very thick,color=black,>=latex,->] (1,3) -- (0.3,1) node[midway, right] {$\mathbf{e}$ };%e
\draw[very thick,color=black,>=latex,->] (-2,0) -- (2.9,0) node[midway,below right] {$\mathbf{f}$ };%f
\draw[color=black,fill=black] (1,3) circle (0.07)  node[above] {$ P _1 $ };
\draw[color=black,fill=black] (-2,0) circle (0.07)  node[below] {$ P_2 $ };
\draw[color=black,fill=black] (0.3,1) circle (0.07)  node[below] {$ P _3 $ };
\draw[color=black,fill=black] (3,0) circle (0.07)  node[below] {$ P _4 $ };
%\draw (1,3) arc (240:275:1cm);
%\draw[thick,red] ([shift=(225:0.8cm)]1,3) arc (225:305:0.8cm);
%\draw[thick,red] ([shift=(35:0.8cm)]0,-2) arc (35:135:0.8cm);

%\draw[very thick] (1.2,2.3) node[below] {$ \alpha $ };
%\draw[very thick] (0.5,-0.8) node[below] {$ \beta$ };

\end{tikzpicture}
%\end{center}
\caption{A  planar concave configuration with $ P _3 $ in the convex hull of the remaining points.\label{fig:concave}}
\end{subfigure}
\hskip 0.5 cm
\centering
\begin{subfigure}{.45\textwidth}
    \centering
    %\begin{center}
\begin{tikzpicture}[thick]%,scale=0.8, every node/.style={scale=1}]
\clip (-2.5,-2.5) rectangle (3.5,3.6);
%\draw[thin,color=white!40!gray]  (-4.9,-4.9) grid (4.9,4.9);
%\draw[->,very thick] (-5,0) -- (5,0) node[below] {$x$};
%\draw[->,very thick] (0,-5) -- (0,5) node[above] {$y$};
%\draw[very thick,color=blue] (-2,0) -- (0,3) node[below] {};
%\draw[very thick,color=blue] (-4,3) -- (-2,0) node[below] {};
%\draw[color=blue,fill=blue] (-4,3) circle (0.1);
%\draw[color=blue,fill=blue] (0,3) circle (0.1);
%\draw[color=blue,fill=blue] (-2,0) circle (0.1);
\draw[very thick,color=black,>=latex,->] (1,3) --   (-2,0)node[midway,above] {$ \mathbf{a}$  };%; node[below] {};%a
\draw[very thick,color=black,>=latex,->] (-2,0) -- (0,-2) node[midway,below] {$\mathbf{b}$ };%b
\draw[very thick,color=black,>=latex,<-] (3,0) -- (0,-2) node[midway,below] {$\mathbf{c}$ };%c

\draw[very thick,color=black,>=latex,<-] (1,3) -- (3,0) node[midway,above] {$\mathbf{d}$ };%d
\draw[very thick,color=black,>=latex,->] (1,3) -- (0,-2) node[midway,above left] {$\mathbf{e}$ };%e
\draw[very thick,color=black,>=latex,->] (-2,0) -- (2.9,0) node[midway,below  right] {$\mathbf{f}$ };%f
\draw[color=black,fill=black] (1,3) circle (0.07)  node[above] {$ P _1 $ };
\draw[color=black,fill=black] (-2,0) circle (0.07)  node[below] {$ P_2 $ };
\draw[color=black,fill=black] (0,-2) circle (0.07)  node[below] {$ P _3 $ };
\draw[color=black,fill=black] (3,0) circle (0.07)  node[below] {$ P _4 $ };
%\draw (1,3) arc (240:275:1cm);
%\draw[thick,red] ([shift=(225:0.8cm)]1,3) arc (225:305:0.8cm);
%\draw[thick,red] ([shift=(35:0.8cm)]0,-2) arc (35:135:0.8cm);

%\draw[very thick] (1.2,2.3) node[below] {$ \alpha $ };
%\draw[very thick] (0.5,-0.8) node[below] {$ \beta$ };

\end{tikzpicture}
%\end{center}
\caption{A  standard planar convex configuration.\label{fig:convex}}
\end{subfigure}
\end{figure}

In the case of a planar configuration $ \Delta $ can be interpreted as the area of the quadrilateral with diagonals $ \mathbf{e} $ and $ \mathbf{f} $, and one obtains the classical Bretschneider's formula.
\begin{corollary}[Bretschneider's formula]\label{cor:bret}
Suppose   $ P _1,P _2, P _3$, and $P _4 $ form   a standard planar convex configuration,  then the area   $A$  of the associated  quadrilateral   is 
\[
    A =  \Delta =\frac{1}{2} \sqrt{   e ^2  f ^2 - \frac{1}{4} (b ^2 + d ^2 - a ^2 - c ^2 ) ^2  }.
\]
\end{corollary} 
\begin{proof}
    Since the area of  a standard  quadrilateral is given by $ A = \frac{1}{2} \| \mathbf{e} \times \mathbf{f} \| =  \Delta $ the proof result follows immediately.
\end{proof}

\begin{remark} 
 Note that the vector formula for the area of a quadrilateral, in our notation  $ A = \frac{1}{2} \| \mathbf{e} \times \mathbf{f} \|$, also holds for concave quadrilaterals with one of the points  in the convex hull of the remaining three points. This is because the area of the Varignon parallelogram (i.e., the parallelogram formed when the midpoints of adjacent sides of a quadrilateral are joined) is one half the area of the quadrilateral. This is true not only for convex quadrilaterals, but also for concave ones. Since the area of the Varignon parallelogram is $ \| \mathbf{e} \times \mathbf{f} \|$ one immediately sees that the area of the quadrilateral is $ A = \frac{1}{2} \| \mathbf{e} \times \mathbf{f} \|$. A consequence of this is that Bretschneider's formula also holds for the concave quadrilaterals. 
 \end{remark}

The  following  condition was used in \cite{saari2005collisions} to derive central configuration equations in terms of angles  for the concave planar four body problem 

\begin{condition}
Let $ P _1,P _2, P _3, P _4 $ be four points in $  \mathbb{R}^3$, the configuration determined by these points is   planar and concave if and only the sum of three of   the  areas of triangle you can make with the four points  add up to the area of the fourth triangle.  
%of the four triangles you can make  add up to zero.
\end{condition}
\begin{proof}
Clearly we have that
\begin{align*}
    \mathbf{a} \times \mathbf{b}&  =- (\mathbf{b} + \mathbf{c}+ \mathbf{d} ) \times \mathbf{b}=
    \mathbf{b} \times \mathbf{c} + \mathbf{b}  \times \mathbf{d} \\  
   % \mathbf{a} \times (\mathbf{b} + \mathbf{c}) + \mathbf{b} \times (\mathbf{b} + \mathbf{c})  \\ 
    & = \mathbf{b} \times \mathbf{c}  -(\mathbf{a} + \mathbf{c} +  \mathbf{d})  \times \mathbf{d}\\
   &  = \mathbf{b} \times \mathbf{c} +  \mathbf{d} \times \mathbf{a} + \mathbf{d} \times \mathbf{c}   . 
\end{align*}
Hence, using the triangle inequality twice, yields 
\begin{align*}
    \| \mathbf{a} \times \mathbf{b} \| &  = 
    \|  \mathbf{b} \times \mathbf{c} +  \mathbf{d} \times \mathbf{a} + \mathbf{d} \times \mathbf{c}  \| \\
   & \leq \|  \mathbf{b} \times \mathbf{c} \|  +  \| \mathbf{d} \times \mathbf{a} + \mathbf{d} \times \mathbf{c}  \| \\
  & \leq  \|  \mathbf{b} \times \mathbf{c} \|  +  \| \mathbf{d} \times \mathbf{a} \| +\|  \mathbf{d} \times \mathbf{c}  \|. 
\end{align*}
Equality holds if and only if $ \mathbf{b} \times \mathbf{c} $, $ \mathbf{d} \times \mathbf{a} $, $ \mathbf{d} \times \mathbf{c} $ and $ \mathbf{a} \times \mathbf{b} $ are collinear. It follows that the four points points are coplanar. 
Let $ A _i $ be the area of the triangle of the triangle containing all points  $ P _j $ with 
$ j \neq i $. Then the condition above can be expressed as 
\[A _4 = A _1 + A _2 + A _3. \]
For a planar configuration this condition holds if and only if the configuration is concave with $ P _4 $ in the convex hull of the other four points. 
Similar conditions hold if one starts from  $ \mathbf{b} \times \mathbf{c} $, $ \mathbf{d} \times \mathbf{a} $, or  $ \mathbf{d} \times \mathbf{c} $
\end{proof}

We now introduce and prove a lesser known planarity condition, which was hinted at in \cite{saari2005collisions}. 
In Section 3 we will use this condition to find new equations for central configurations in terms of angles. 
%\begin{theorem}[1st planarity condition]
\begin{condition}\label{cond:planarity4}
With the notations above let   $ A _1 = \frac{1}{2} \| \mathbf{b} \times \mathbf{c}\|  \neq 0 $, $ A _2 = \frac{1}{2} \| \mathbf{c} \times \mathbf{d} \| \neq 0$, $ A _3 = \frac{1}{2}  \| \mathbf{d} \times \mathbf{a} \| \neq 0$ and $ A _4 = \frac{1}{2} \| \mathbf{a} \times \mathbf{b}\|  \neq 0 $, then 
\begin{enumerate}
         \item   The configuration is a  standard planar  convex or concave (with either $ P _1 $ or $ P _3 $ in the convex hull of the other points) configuration if and only if $ \Delta = A _2 + A _4 $. 
    \item The configuration is  a standard planar  convex  or concave (with either $ P _2 $ or $ P _4$ in the convex hull of the other points) configuration if and only if $ \Delta = A _1 + A _3 $.
      
   \item  The configuration is a standard convex planar configuration if and only if $ \Delta = A _1 + A _3 $ and $ \Delta = A _2 + A _4 $.
\end{enumerate} 
%$ \Delta ^2 =  (A _1 + A _2) ^2 $  and $ \Delta ^2 =  (A _3 + A _4) ^2 $ if and only if the configuration is a standard convex planar configuration.  
\end{condition}
%\end{theorem} 

\begin{proof} 
Clearly we have that
\begin{align*}
    \mathbf{e} \times \mathbf{f}&  = (\mathbf{a} + \mathbf{b}) \times (\mathbf{b} + \mathbf{c})=
   (\mathbf{a} + \mathbf{b}) \times \mathbf{b} + (\mathbf{a} + \mathbf{b}) \times \mathbf{c} \\  
   % \mathbf{a} \times (\mathbf{b} + \mathbf{c}) + \mathbf{b} \times (\mathbf{b} + \mathbf{c})  \\ 
    & = \mathbf{a} \times \mathbf{b}  -(\mathbf{c} + \mathbf{d})  \times \mathbf{c} = \mathbf{a} \times \mathbf{b} +  \mathbf{c} \times \mathbf{d}  . 
\end{align*} 

Using the expression for $  \mathbf{e} \times \mathbf{f}  $ above yields,  
\begin{align*}
    4\Delta ^2&  = \| \mathbf{e} \times \mathbf{f} \| ^2 = \| \mathbf{a} \times \mathbf{b} + \mathbf{c} \times \mathbf{d}\| ^2\\
    &     =  \| \mathbf{a} \times \mathbf{b}\| ^2  +\|  \mathbf{c} \times \mathbf{d}\| ^2 + 2 (\mathbf{a} \times \mathbf{b}) \cdot  (\mathbf{c} \times \mathbf{d} ) \\
  & \leq   \| \mathbf{a} \times \mathbf{b}\| ^2  +\|  \mathbf{c} \times \mathbf{d}\| ^2 + 2| (\mathbf{a} \times \mathbf{b}) \cdot  (\mathbf{c} \times \mathbf{d} ) |\\
  &  \leq \| \mathbf{a} \times \mathbf{b}\| ^2  +\|  \mathbf{c} \times \mathbf{d}\| ^2 + 2 \| \mathbf{a} \times \mathbf{b}) \| \|  (\mathbf{c} \times \mathbf{d} )\| \\
  & = ( \| \mathbf{a} \times \mathbf{b}\|   +\|  \mathbf{c} \times \mathbf{d}\|)^2=(A _2 + A _4) ^2  
 \end{align*} 
where the last inequality follows from the Cauchy-Schwarz inequality. 
If $ \Delta = (A _2 + A _4) $ then the Cauchy-Schwarz inequality reduces to an equality. Since $ \| \mathbf{a} \times \mathbf{b} \| \neq 0 $ and $\|  \mathbf{c} \times \mathbf{d} \| \neq 0 $ the equality holds when $ \mathbf{a} \times \mathbf{b} $ and $ \mathbf{c} \times \mathbf{d} $ are parallel. In this case the configuration is planar. Since $ A _1 , A _2 , A _3 , A _4 \neq 0 $, it follows that the configuration is also non-degenerate.  Moreover, we must have  $ (\mathbf{a} \times \mathbf{b}) \cdot  (\mathbf{c} \times \mathbf{d} ) =|(\mathbf{a} \times \mathbf{b}) \cdot  (\mathbf{c} \times \mathbf{d} ) | $. Consequently $ \mathbf{a} \times \mathbf{b} $ and $ \mathbf{c} \times \mathbf{d} $ point in the same direction. This is enough to exclude concave configurations with $ P _2 $ or $ P _4 $ in the convex hull of the remaining masses.  
%This proves part 1.

Similar computations allow us to exclude concave configurations with $ P _1 $ or $ P _3 $ in the convex hull of the remaining masses. We have that

\begin{align*}
    \mathbf{e} \times \mathbf{f}&  = (\mathbf{a} + \mathbf{b}) \times (\mathbf{b} + \mathbf{c})=
    \mathbf{a} \times (\mathbf{b} + \mathbf{c})  + \mathbf{b}  \times (\mathbf{b} + \mathbf{c})  \\   
    & = - \mathbf{a} \times (\mathbf{a} + \mathbf{d})  + \mathbf{b}  \times  \mathbf{c}
    =  \mathbf{d} \times  \mathbf{a}  + \mathbf{b}  \times  \mathbf{c}
. 
\end{align*} 

Using the expression for $  \mathbf{e} \times \mathbf{f}  $ above yields,  
\begin{align*}
    \Delta ^2&  = \| \mathbf{e} \times \mathbf{f} \| ^2 = \| \mathbf{d} \times \mathbf{a} + \mathbf{b} \times \mathbf{c}\| ^2\\
    &     =  \| \mathbf{d} \times \mathbf{a}\| ^2  +\|  \mathbf{b} \times \mathbf{c}\| ^2 + 2 (\mathbf{d} \times \mathbf{a}) \cdot  (\mathbf{b} \times \mathbf{c} ) \\
  & \leq   \| \mathbf{d} \times \mathbf{a}\| ^2  +\|  \mathbf{b} \times \mathbf{c}\| ^2 + 2| (\mathbf{d} \times \mathbf{a}) \cdot  (\mathbf{b} \times \mathbf{c} ) |\\
  &  \leq \| \mathbf{d} \times \mathbf{a}\| ^2  +\|  \mathbf{b} \times \mathbf{c}\| ^2 + 2 \| \mathbf{d} \times \mathbf{a}) \| \|  (\mathbf{b} \times \mathbf{c} )\| \\
  & = ( \| \mathbf{d} \times \mathbf{a}\|   +\|  \mathbf{b} \times \mathbf{c}\|)^2=(A _1 + A _3) ^2.
 \end{align*} 
If $ \Delta = (A _1 + A _3) $ then  $  \mathbf{d} \times \mathbf{a} $ and $ \mathbf{b} \times \mathbf{c} $ are parallel and $ \mathbf{d} \times \mathbf{a} $ and $ \mathbf{b} \times \mathbf{c} $ point in the same direction. This is enough to exclude concave configurations with $ P _1 $ or $ P _3$ in the convex hull of the remaining masses.  %This proves part 2. 

Combining the arguments above, if $ \Delta = A _1 + A _3$ and $ \Delta = A _2 + A _4 $ then the configuration must be a  standard planar  convex configuration. 

Conversely, suppose that the four points  form a standard  convex quadrilateral, with $e$ and $f$ as diagonals. Then, by Corollary \ref{cor:bret} its area is  $  \Delta $ and it equals  the sum of the areas of the triangles $ A _1 $ and $ A _3 $, or of the triangles  $ A _2 $ and $ A _4 $ .  This concludes the proof.  
\end{proof}
%\begin{remark} 
%Note that  the following two conditions taken separately 
%\begin{enumerate}
%    \item  $ \Delta ^2 =  (A _1 + A _2) ^2 $,
%    \item $ \Delta ^2 =  (A _3 + A _4) ^2 $
%\end{enumerate}
%are good planarity conditions, but they are not enough to ensure that the resulting planar configuration is convex. 
%\end{remark}

We are now ready to introduce the last planarity condition we discuss in this article.  It will be  shown in   Section \ref{sec:cocircular} that  this condition   is  important in connection with the co-circular four body problem (see \cite{cors2012four} for a nice approach to the co-circular four body problem). 
%\begin{theorem}[2nd Planarity Condition]\label{thm:planarity2}
\begin{condition}\label{cond:planarity5}
    Consider a convex configuration of four bodies and  let $ \alpha$ be the angle between the vectors $ \mathbf{a} $ and $ \mathbf{d} $, and $ \beta $ be the angle between $ \mathbf{b} $ and $ \mathbf{c} $. If   $ \gamma = \frac{ \alpha + \beta } { 2  } $, then 
    \begin{equation}\label{eqn:planarity5} 
F =  (a c + b d + ef)  (ac + b d - e f)- 4 abcd \cos ^2 \gamma=0 
 \end{equation} 
if and only if the configuration is a standard convex planar configuration.  
\end{condition} 
\begin{proof} 

Suppose configuration is a standard convex planar configuration, then 
\begin{equation} \label{eqn:planarity_proof}
     \| \mathbf{a} \times \mathbf{d} \| + \| \mathbf{b} \times \mathbf{c} \| = \Delta. 
 \end{equation} 
Since $ \mathbf{f} = \mathbf{b} + \mathbf{c} = - (\mathbf{a} + \mathbf{d})   $ 
we have that 
\begin{align*} 
    \| \mathbf{f} \| ^2 &= (\mathbf{b} + \mathbf{c}) \cdot (\mathbf{b} + \mathbf{c}) = (\mathbf{a} + \mathbf{d}) \cdot (\mathbf{a} + \mathbf{d})\\
    & = \| \mathbf{b} \| ^2 + \| \mathbf{c} \| ^2 + 2 (\mathbf{b} \cdot \mathbf{c} )= \| \mathbf{a} \| ^2 + \| \mathbf{d} \| ^2 + 2( \mathbf{a} \cdot \mathbf{d}),
\end{align*} 
and  thus,
\begin{equation}\label{eqn:proof} \mathbf{b} \cdot \mathbf{c} - \mathbf{a} \cdot \mathbf{d} = \frac{1}{2} (\| \mathbf{a} \| ^2 + \| \mathbf{d} ^2 \|-\| \mathbf{b} \| ^2 - \| \mathbf{c} \| ^2  ) = \frac{1}{2} (a ^2 + d ^2 - b ^2 - c ^2 ) 
\end{equation} 
Squaring and adding equation \eqref{eqn:planarity_proof} and \eqref{eqn:proof} yields:
\begin{equation}\label{eqn:proof2} 
    (\| \mathbf{a} \times \mathbf{d} \| + \| \mathbf{b} \times \mathbf{c} \| )^2 + (\mathbf{b} \cdot \mathbf{c} - \mathbf{a} \cdot \mathbf{d} ) ^2    = \Delta ^2 + \frac{1}{4} (a ^2 + d ^2 - b ^2 - c ^2 ) ^2.
\end{equation} 
We now want to rewrite the left hand side of equation \eqref{eqn:proof2}. Let $ \alpha \in [0, \pi ]$ be the angle between the vectors $ \mathbf{a} $ and $ \mathbf{d} $, and $ \beta \in [0, \pi ] $ be the angle between $ \mathbf{b} $ and $ \mathbf{c} $, then  
\[
\begin{aligned} 
    \| \mathbf{a} \times \mathbf{d} \| & = a d  \sin \alpha \\
    \| \mathbf{b} \times \mathbf{c} \| & = b c   \sin \beta \\
    \mathbf{b} \cdot \mathbf{c} & =  b c \cos  \beta\\
    \mathbf{a} \cdot \mathbf{d} & =  a d \cos  \alpha.
\end{aligned}  
\]

Let  $ \gamma = \frac{ \alpha + \beta } { 2 } $ then,  using these expression, we obtain 
\begin{align*}
& (\| \mathbf{a} \times \mathbf{d} \| + \| \mathbf{b} \times \mathbf{c} \| )^2 + (\mathbf{b} \cdot \mathbf{c} - \mathbf{a} \cdot \mathbf{d} ) ^2  = \\
&  a ^2 d ^2 + b ^2 c ^2 + 2 abcd (\sin \alpha \sin \beta - \cos \alpha \cos \beta)  = \\
&  a ^2 d ^2 + b ^2 c ^2 -2 abcd \cos (\alpha + \beta) =  \\
&   (a  d  + b  c) ^2 - 2abcd( 1 + \cos (\alpha + \beta))= \\
&   (a d + b c) ^2 -4abcd \cos ^2 \gamma. 
\end{align*}
Substituting the previous equation into \eqref{eqn:proof2}, yields
 
\[\Delta ^2  = (a d + b c) ^2 - \frac{1}{4} (a ^2 + d ^2 - b ^2 - c ^2) ^2 -4 abcd \cos ^2 \gamma   \]
or 
\[
    \frac{\Delta ^2}{4} = (s - a) (s - b) (s - c) (s - d) - abcd \cos ^2 \gamma.     
\]
Comparing this with the equation in Lemma \ref{lemma:ef} establishes the planarity condition.

Conversely, if \eqref{eqn:planarity5} holds, following the previous reasoning backwards we find 
 \eqref{eqn:planarity_proof}, from which it follows that the configuration is coplanar and it is either convex or concave with $ P _1 $ or $  P _3 $ in the convex hull of the remaining points. The proof follows since the configuration is convex by hypothesis.  

 \end{proof}
Incidentally, we can use the previous theorem to give a proof of Ptolemy's inequality. 
 \begin{corollary}[Ptolemy's Inequality]
     Suppose the points $ P _1 , \ldots , P _4 $ form a standard planar convex quadrilateral, then 
     \[ac + b d - e f \geq 0 \]
     with equality if and only if the quadrilateral is cyclic. 
 \end{corollary}  

 \begin{proof}
     If  the points $ P _1 , \ldots , P _4 $ form a standard convex quadrilateral, then by Theorem \ref{cond:planarity5} we have 

     \[ (ac + b d - e f)= 4 \frac{ abcd \cos ^2 \gamma } {(a c + b d + ef)  }\geq 0. \]
     Let $ \tilde \alpha = \pi - \alpha $, and $ \tilde \beta = \pi - \beta $, the interior angles of the quadrilateral corresponding to $ \alpha $  and $ \beta $, respectively. Recall that the quadrilateral is cyclic if and only if $ \tilde \alpha + \tilde \beta = \pi $. If $ \tilde \gamma = \frac{ \tilde \alpha + \tilde \beta } { 2}$, then   $  \gamma = \frac{ \pi } { 2 }   $, which implies $ \cos \gamma = 0 $. This completes the proof. 

 \end{proof} 

 Note that the theorem also holds in the case the quadrilateral is concave with either $ P _1 $ or $ P _3 $ in the convex hull of the other points. 
%%%%%%%%%%%%%%%%%%%%%%%%%%%%%%%%%
\section{Planarity Condition 4 and c.c equations in terms of angles}
%%%%%%%%%%%%%%%%%%%%%%%%%%%%%%%%%%%%%

From  Planarity Condition \ref{cond:planarity4}, it follows that if we are looking for   planar central configurations      we can impose one of the following two conditions:
\begin{itemize} 
    \item $F_1=\Delta-A_1-A_3=0$ 
    \item $F_2=\Delta-A_2-A_4=0$.
\end{itemize} 

In order to find the extrema of $U$ under  the constraints $I-I_0=0$,  and one of $F_1=0$ and $F_2=0$, let $\lambda,\eta_1$ and $\eta_2$ be Lagrange multipliers, so that we have to find the extrema of one of the two equations 
\begin{equation}
%U+\bar\lambda(I-I_0)+\eta_1 F_1+\eta_2F_2
U+\lambda  M(I-I_0)+\eta_k F_k, \quad k = 1,2 
\label{conditionSaari}
\end{equation}
If we choose to use the constraint $ F _1 = 0 $ the condition for a planar extrema is  
%m_im_j\left( \lambda -r_{ij}^{-3}\right)r_{ij}+\eta_1\frac{\partial F_1}{\partial r_{ij}}+\eta_2\frac{\partial F_2}{\partial r_{ij}}=0, \qquad  1\leq i<j\leq 4
\begin{equation}\label{eqn:F1} 
m_im_j\left( \lambda -r_{ij}^{-3}\right)r_{ij}+\eta_1\frac{\partial F_1}{\partial r_{ij}}=0, \qquad  1\leq i<j\leq 4
\end{equation} 
\[I-I_0=0, \qquad F_1 =0,\]
%where $\lambda=(2/m')\bar\lambda$
and if we choose to use $ F _2 = 0 $ it is 
\begin{equation} \label{eqn:F2}
m_im_j\left( \lambda -r_{ij}^{-3}\right)r_{ij}+\eta_2\frac{\partial F_2}{\partial r_{ij}}=0, \qquad  1\leq i<j\leq 4
\end{equation} 
\[I-I_0=0, \qquad F_2 =0.\]

%\begin{remark}
%Observe that the Lagrange multiplier $\bar\lambda$  has the same values on planar (convex)  central configurations  in equation (\ref{eqdef}),(\ref{conditionDziobek}) and (\ref{conditionSaari}). Namely $\bar\lambda=\frac{U(q)}{2(q^TMq)}$. This follows from Euler's theorem on homogeneous functions and the equations $S=0$ and $F_1=F_2=0$. Consequently $\lambda$ has the same value in equation (\ref{conditionDziobek}) and (\ref{conditionSaari}) when computed at a  convex central configurations. 
%\end{remark}

 Note that one can also impose both the conditions  $ F _1 = 0 $ and $ F _2 = 0 $, in which case the solutions of the variational problem will be a standard planar convex configurations. 

To compute the partials of $F _k $ ($ k = 1,2 $)   we will use the expression of $\Delta$  given in Corollary \ref{cor:bret}, and expressions of the areas  $A_1,A_2,A_3$ and $A 
_4 $ obtained with  Heron's formula  

\[\begin{split}
A_1&=\frac 1 4 \sqrt{2(b^2c^2+b^2f^2+c^2f^2)-(b^4+c^4+f^4)}\\
A_2&=\frac 1 4 \sqrt{2(c^2d^2+c^2e^2+d^2e^2)-(c^4+d^4+e^4)}\\
A_3&=\frac 1 4 \sqrt{2(a^2d^2+a^2f^2+d^2f^2)-(a^4+d^4+f^4)}\\
A_4&=\frac 1 4 \sqrt{2(a^2b^2+a^2e^2+b^2e^2)-(a^4+b^4+e^4)}.
%A_1&=\frac 1 4 \sqrt{2(a^2d^2+a^2e^2+d^2e^2)-(a^4+d^4+e^4)}\\
%A_2&=\frac 1 4 \sqrt{2(b^2c^2+b^2e^2+c^2e^2)-(b^4+c^4+e^4)}\\
%A_3&=\frac 1 4 \sqrt{2(a^2b^2+a^2f^2+b^2f^2)-(a^4+b^4+f^4)}\\
%A_4&=\frac 1 4 \sqrt{2(c^2d^2+c^2f^2+d^2f^2)-(c^4+d^4+f^4)}
\end{split}
\]

%%%%%%%%%%%%%%%%%%%%%%%%%%%%%%%%%%%%%%%%%%%%%%%%%%%%%%%%%%%%%%%
%\iffalse
 Let $ \theta $ be the angle between the diagonal, then  the partial derivatives of $\Delta$ with respect to $e$ and $f$ are:
\[
\frac{\partial \Delta}{\partial e}=\frac{ef^2}{4\Delta}=\frac {f }{2 \sin\theta}=\frac f 2 \csc\theta \qquad \frac{\partial \Delta}{\partial f}=\frac{fe^2}{4\Delta}=\frac {e}{ 2 \sin\theta}=\frac e 2 \csc\theta
\]
where we used that $\Delta=\frac 1 2ef \sin\theta$, an equation for the area of a quadrilateral (see \cite{hobson2004treatise}).
The partials with respect to the remaining mutual distances are
\[\begin{split}
\frac{\partial \Delta}{\partial a} &=\frac{a}{8\Delta}(b^2+d^2-a^2-c^2)=\frac a 2 \cot\theta\\
\frac{\partial \Delta}{\partial b} &=-\frac{b}{8\Delta}(b^2+d^2-a^2-c^2)=-\frac b 2 \cot\theta\\
\frac{\partial \Delta}{\partial c} &=\frac{c}{8\Delta}(b^2+d^2-a^2-c^2)=\frac c 2 \cot\theta\\
\frac{\partial \Delta}{\partial d} &=-\frac{d}{8\Delta}(b^2+d^2-a^2-c^2)=-\frac d 2 \cot\theta,
\end{split}
\]
where we used that $\Delta=\frac 1 4(b^2+d^2-a^2-c^2)\tan\theta$,  another formula for the  area of a quadrilateral (see \cite{hobson2004treatise}).

The partials of $A_1$ and $A_2$ can be computed in the following way. Given a triangle of sides $\alpha,\beta$ and $\gamma$ Heron's formula for the area gives
\[A=\frac 1 4\sqrt{2(\alpha^2\beta^2+\alpha^2\gamma^2+\beta^2\gamma^2)-(\alpha^4+\beta^4+\gamma^4)}.\]
Consequently, by symmetry and the law of cosines,  all the partials of $A$ with respect to the mutual distances have the form
\[
\frac{\partial A}{\partial \alpha} =\frac{\alpha}{8A}(\beta^2+\gamma^2-\alpha^2)=\frac{\alpha}{8A}(2\beta\gamma\cos\delta)
\]
where $\delta$ is the angle between the sides $\beta$ and $\gamma$. As $A=\frac 1 2 \beta\gamma\sin\delta$  we find 
\[
\frac{\partial A}{\partial \alpha}=\frac 1 2\alpha\cot\delta.
\]
Now let $\theta_{ijk}$ be the angle formed by the vertices $\{i,j,k\}$.
From equation \eqref{eqn:F1}, putting all the computations above together, and absorbing the $\frac 1 2$ multiples in the Lagrange multipliers $\eta_1$,  we obtain  the  six equations
\begin{align} 
  r_{12}^{-3}&=\lambda+\frac{\eta _1 }{m_1m_2}[\cot\theta- \cot \theta _{ 142 } ]\label{eqn:angle1}\\
r_{13}^{-3}&=\lambda+\frac{\eta _1 }{m_1m_3}\left[ \frac{ r _{ 24 } } { r _{ 13 } } \csc\theta \right]\label{eqn:angle2}\\
r_{14}^{-3}&=\lambda+\frac{\eta _1 }{m_1m_4}\left[- \cot \theta -\cot \theta_{124}\right]\label{eqn:angle3}\\
r_{23}^{-3}&=\lambda+\frac{\eta _1 }{m_2m_3}\left[-\cot \theta - \cot \theta _{243}\right]\label{eqn:angle4}\\
r_{24}^{-3}&=\lambda+\frac{\eta _1 }{m_2m_4}\left[\frac{ r _{ 13 }} { r _{ 24 } } \csc \theta - \cot \theta _{ 23 4 }- \cot \theta _{ 214 }  \right]\label{eqn:angle5}\\
r_{34}^{-3}&=\lambda+\frac{\eta _1 }{m_3m_4}[\cot \theta - \cot \theta _{ 324 }   ],\label{eqn:angle6}
\end{align} 
together with $I-I_0=0$ and $F _1 =0$.
Similarly, from equation \eqref{eqn:F2} we obtain 
\begin{align} 
  r_{12}^{-3}&=\lambda+\frac{\eta _2 }{m_1m_2}[\cot\theta- \cot \theta _{ 132 } ]\label{eqn:angle1a}\\
r_{13}^{-3}&=\lambda+\frac{\eta _2 }{m_1m_3}\left[ \frac{ r _{ 24 } } { r _{ 13 } } \csc\theta - \cot \theta _{ 143 }- \cot \theta _{ 123 }  \right]\label{eqn:angle2a}\\
r_{14}^{-3}&=\lambda+\frac{\eta _2 }{m_1m_4}\left[- \cot \theta -\cot \theta_{134}\right]\label{eqn:angle3a}\\
r_{23}^{-3}&=\lambda+\frac{\eta _2 }{m_2m_3}\left[-\cot \theta - \cot \theta _{213}\right]\label{eqn:angle4a}\\
r_{24}^{-3}&=\lambda+\frac{\eta _2 }{m_2m_4}\left[\frac{ r _{ 13 }} { r _{ 24 } } \csc \theta \right]\label{eqn:angle5a}\\
r_{34}^{-3}&=\lambda+\frac{\eta _2 }{m_3m_4}[\cot \theta - \cot \theta _{ 314 }   ],\label{eqn:angle6a}
\end{align} 
together with $I-I_0=0$ and $F _1 =0$.

We also mention the well-known  relation of Dziobek \cite{dziobek1900uber}
\begin{equation}\label{eqn:dziobek}
   (r _{ 12 }^{ - 3 }  - \lambda) (r _{ 34 } ^{ - 3 } - \lambda) = (r _{ 13 } ^{  - 3 } - \lambda) (r _{ 24 } ^{ - 3 } - \lambda) = (r _{ 14 } ^{ - 3 } - \lambda) (r _{ 23 } ^{ - 3 } - \lambda)        
\end{equation} 
which is required of any planar 4-body central configuration (see \cite{schmidt2002central} for a derivation). 

\section{Some applications of the new equations}
 We start by recalling the following well known lemma  
 
 \begin{lemma}\label{lem:diag_sides}
     In  a planar  convex central configuration all exterior sides are shorter than the diagonals, and that all the exterior sides are greater than or equal to $ 1/\sqrt[3]{ \lambda } $, and the lengths of all diagonals are greater than or equal to $ 1/\sqrt[3]{ \lambda } $.
 \end{lemma} 
 See \cite{schmidt2002central}  for a proof, and \cite{corbera2016four} for an analogous result for point vortices.
 A consequence of the lemma above is the following useful fact:

\begin{lemma} \label{lem:eta} For any  standard convex planar central configuration $ \eta _1 <0 $. 
\end{lemma} 
\begin{proof} 
    From \eqref{eqn:angle2} we obtain 
    \[
        \frac{ 1 } { r _{ 13 } ^3 } - \lambda = \frac{ \eta _1 } {  m _1 m _3 } \frac{ r _{ 24 } } { r _{ 13 } } \csc \theta.
    \]  
    For any standard convex planar central configuration the left hand side of the equation above  is negative by the previous lemma.  Moreover,  $ \csc \theta >0 $ since $ \theta >0 $. It follows that $ \eta _1 <0 $.
\end{proof} 

We can now give a simple proof of the following proposition (Lemma 3.2 in \cite{cors2012four}) dealing with  co-circular configurations.  
%=======================================
\begin{proposition}\label{prop:trapezoid}
   Consider four bodies forming a concave or a standard convex configuration.  
   If $ m _1 = m _2 $, and $ m _3 = m _4 $, then the corresponding co-circular central configuration must be an isosceles trapezoid.  
\end{proposition} 
%========================================
%----------------------------------------
\begin{proof}
Since the quadrilateral determined by the masses is inscribed in a circle, the angles  $ \theta _{ 234 } $ and $ \theta _{ 214 } $ are supplementary. It follows that $  \cot \theta _{ 23 4 }+ \cot \theta _{ 214 } =0 $. If  $ m _1 = m _2 $, and $ m _3 = m _4 $, subtracting equation \eqref{eqn:angle5} from \eqref{eqn:angle2} yields 
\[
    \frac{ r _{ 24 }  ^3 - r _{ 13 }^3 } { r _{ 13 } ^3 r _{ 24 } ^3 } =  \frac{ \eta _1 } { m _1 m _3 } \frac{ r _{ 24 } ^2 - r _{ 13 } ^2 } { r _{ 13 } r _{ 24 } } \csc \theta. 
\]
Suppose $ r _{ 24 } \neq r _{ 13 } $ , then, since $ \csc \theta >0 $ and $ \eta _1 <0 $ by Lemma \ref{lem:eta}, 
the left and right side of the equation have opposite signs. Consequently,  we must have $ r _{ 13 } = r _{ 24 } $. Since the configuration is on a circle, it follows that $r_{14}=r_{23} $, 
A similar reasoning can be repeated with equation \eqref{eqn:angle5a} from \eqref{eqn:angle2a}, this will exclude the remaining concave cases. It follows that  the configuration is an isosceles trapezoid.
\end{proof} 
%-------------------------------
We now use the c.c. equations in terms of angles in order to capture an aspect of   Lemma 2.5 in \cite{hampton2014relative}, see also \cite{deng2017four} for a different proof of a similar result.
%===========================
\begin{proposition} 
Suppose  that  we  have  a  standard convex planar  central configuration  with $ m _1 = m _2 $, and $ m _3 = m _4 $, then 
\begin{equation} \label{eqn:symmetrylemma}
    r _{ 13 } = r _{ 24 }\quad \mbox{ if an only if }\quad r _{ 14 } = r _{ 23 }.
\end{equation} 
In  this  case,  the  configuration  is    an  isosceles  trapezoid  with  bodies $1$ and $2$ on one base,  and $3$ and $4$ on the other.
%If either equation in \eqref{eqn:symmetrylemma} holds, the configuration is convex.
% the corresponding  configuration is co-circular. %must be an isosceles trapezoid.  
\end{proposition}
%============================
\begin{proof}
   Suppose $ r _{ 13 } = r _{ 24 } $, then the quadrilateral is equidiagonal and 
    \begin{align*} 0& =r_{13} ^{-3}-r _{ 24 } ^{-3} =\frac{\eta _1 }{m_1m_3}\left[ \frac{ r _{ 24 } } { r _{ 13 } } \csc\theta \right]
-\frac{\eta _1 }{m_2m_4}\left[\frac{ r _{ 13 }} { r _{ 24 } } \csc \theta - \cot \theta _{ 23 4 }- \cot \theta _{ 214 }  \right]\\
& = \frac{\eta _1 }{m_2m_4}\left[ \cot \theta _{ 23 4 }+ \cot \theta _{ 214 }  \right].
\end{align*} 
Hence, $ \cot \theta _{ 23 4 }+ \cot \theta _{ 214 } =0 $ and $ \theta _{ 234 } $ and $ \theta _{ 214 } $ are supplementary. The configuration is co-circular.
However, if a  cyclic quadrilateral is also equidiagonal, it is an isosceles trapezoid.
This can be proved as follows. By Ptolemy's second theorem 
\[
    \frac{ e } { f } = \frac{ a d +  b c } { a b + c d } .
\]
Since $ e = f $, then $ a d +  b c = a b + c d $, which imply that $ (a - c) (b - d) = 0  $.  The last equality has the two possible solutions $a=c $ and $ b = d $. Moreover, if in a  cyclic quadrilateral  a pair of opposite  sides is congruent then the other sides must be parallel. For instance, if $b  = d $, since  equal chords subtend equal angles at the circumference of the circle, we have that $ \theta _{ 213 } = \theta _{ 134 } $. Consequently, $ \mathbf{c} $ is parallel to $ \mathbf{a} $. The  case  $a=c $ is similar. It follows that the quadrilateral is a rhombus or an isosceles trapezoid with $ r _{ 14 } = r _{ 23 } $. But the only cyclic rhombus is a square. Thus, in both cases  $ b = d $, that is  $ r _{ 14 } = r _{ 23 } $.

%In a co-circular configuration  the angle between a side and a diagonal is equal to the angle between the opposite side and the other diagonal, and thus  $ \theta _{ 124 } = \theta _{ 243 } $.

%It follows that
%\begin{align*} 
%r_{14}^{-3}-r_{23}^{-3}&=-\frac{\eta _1 }{m_1m_4}\left[- \cot \theta -\cot \theta_{124}\right]+\frac{\eta _1 }{m_2m_3}\left[-\cot \theta - \cot \theta _{243}\right] \\
%& = \frac{\eta _1 }{m_1m_4}[\cot \theta_{124}-\cot \theta _{243}]\\
%    & = 0,
%\end{align*}
%and thus $ r _{ 14 } = r _{ 23 } $.

Conversely assume that $ r _{ 14 } = r _{ 23 } $, then 
\begin{align*} 
0=r_{14}^{-3}-r_{23}^{-3}= \frac{\eta _1 }{m_1m_4}[-\cot \theta_{124}+\cot \theta _{243}]\\
\end{align*}
and $  \theta_{124}= \theta _{243} $. It follows that two of the opposite sides are parallel, that is, $ \mathbf{a} $ is parallel to $ \mathbf{c} $. Hence, the configuration is either an isosceles trapezoid  or a parallelogram.

If it is an isosceles trapezoid then the diagonals have equal length, that is, $ r _{ 13 } = r _{ 24 } $, and we are done. 

If it is a parallelogram opposite sides are equal in length, that is, $ r _{ 14 } = r _{ 23 } $, and  $ r _{ 12 } = r _{ 34 } $. From equation \eqref{eqn:dziobek} and Lemma \ref{lem:diag_sides} it follows that all the external sides must have equal length. The quadrilateral is then  a rhombus. 
Moreover, opposite sides are parallel, and hence $  \theta _{ 142 } =\theta _{ 324 } $. From equation
%Moreover, opposite angles are equal in measure, thus  $ \theta _{ 123 } = \theta _{ 143 } $.  Since  $  \theta_{124}= \theta _{243} $ it follows  that $  \theta _{ 142 } =\theta _{ 324 } $. From equation
\eqref{eqn:angle1} and \eqref{eqn:angle6} we obtain
\[
  0 =  r _{ 12 } ^{ - 3 } - r _{ 34 } ^{ - 3 } = \eta _1 
  \left[ \frac{ 1 } { m _1  ^2} - \frac{ 1 } { m _3  ^2} 
  \right] (\cot \theta - \cot \theta _{ 142 } ).
\]
 Since $ \cot \theta - \cot \theta _{ 142 } \neq 0$, it follows that $  m _1 = m _2 = m _3 = m _4 $. In this case the rhombus reduces to a square (this follows from the uniqueness of rhombus configurations, see  \cite{long2002four} for a proof), and thus $ r _{ 13 } = r _{ 24 } $.

%Since the configuration is co-circular and $ m _1 = m _2 $ and $ m _3 = m _4 $ by Proposition \ref{prop:trapezoid} the configuration is an isosceles trapezoid. Hence, the diagonals have equal length, that is, $ r _{ 13 } = r _{ 24 } $.
% \begin{align*} r_{13} ^{-3}-r _{ 24 } ^{-3}  = -\frac{\eta _1 }{m_2m_4}\left[ \cot \theta _{ 23 4 }+ \cot \theta _{ 214 }  \right]=0
%\end{align*} 
%since $  \cot \theta _{ 23 4 }+ \cot \theta _{ 214 } $ in a co-circular configuration.%, and hence it is convex.

%In both cases we have that $ r _{ 23 } = r _{ 14 } $,  $ m _1 = m _2 $ and $ m _3 = m _4 $. It follows by equation \eqref{eqn:angle3} and \eqref{eqn:angle4} that  $ \theta _{ 243 } = \theta _{ 124 } $, and thus the opposite sides of length $ r _{ 12 } $ and $ r _{ 34 } $ are parallel.
%Consequently the quadrilateral is a trapezoid. Moreover, it is an isosceles trapezoid since the diagonals have the same length. 
\end{proof}
%------------------------------------
Another aspect of   Lemma 2.5 in \cite{hampton2014relative}, is captured in the following proposition.
%===================================
\begin{proposition}
    Suppose  that  we  have  a  standard convex planar  central configuration  with $ m _1 = m _3 $, and $ m _2 = m _4 $, then 
\begin{equation} \label{eqn:symmetrylemma1}
    r _{ 12 } = r _{ 34 }\quad \mbox{ if an only if }\quad r _{ 14 } = r _{ 23 }.
\end{equation} 
  In  this  case,  the  configuration  is is a rhombus with bodies $1$ and $3$
opposite each other.
\end{proposition}
%===================================
%--------------------------------------
\begin{proof}%\textcolor{red}{converse is missing in the proof } 
    Suppose that $ m _1 = m _3 $, $ m _2 = m _4 $, and $ r _{ 12 } = r _{ 34 } $, then from equations \eqref{eqn:angle1} and \eqref{eqn:angle6}, we obtain 
    \begin{align*}
       0 = r _{ 12 } ^{ - 3 } - r _{ 34 } ^{ - 3 } = \frac{ \eta _1 } { m _1 m _2 } \left(- \cot \theta _{ 142 } + \cot \theta _{ 324 } \right).  
    \end{align*} 
The only solution of this equation with  $0<\theta _{ 324 }, \theta _{ 142 } <\pi $ is $\theta _{ 324 }= \theta _{ 142 } $.

%Similarly, from equations \eqref{eqn:angle1a} and \eqref{eqn:angle6a}, we obtain 
%    \begin{align*}
%       0 = r _{ 12 } ^{ - 3 } - r _{ 34 } ^{ - 3 } = \frac{ \eta _1 } { m _1 m _2 } \left( \cot \theta _{ 132 } - \cot \theta _{ 314 } \right).
%    \end{align*} 
%The only solution of this equation with  $0<\theta _{ 132 }, \theta _{ 314 } <\pi $ is $\theta _{ 132 }= \theta _{ 314 } $.
From this it follows that 
$\mathbf{d} $ and $\mathbf{b} $ are parallel. Hence, the configuration is either an isosceles trapezoid  or a parallelogram.

If it is an isosceles trapezoid it must be a cyclic quadrilateral (and thus $ \cot \theta _{ 234 } + \cot \theta _{ 214 } = 0 $) and have $ r _{ 12 } = r _{ 34 } $ and have diagonals of equal length, that is, 
$ r _{ 13 } = r _{ 24 } $.
Then, from equation \eqref{eqn:angle2} and \eqref{eqn:angle5}
\begin{equation} \label{eqn:masses} 
    0 = r _{ 13 } ^{ - 3 } - r _{ 24 } ^{ - 3 } = \eta _1 \csc \theta \left( \frac{ 1 } { m _1 ^2 } - \frac{ 1 } { m _2 ^2  } \right).
\end{equation} 
Thus $ m _1 = m _2 $ and all the masses are equal. In this case the isosceles trapezoid reduces to a square (this follows, for instance, from the uniqueness of isosceles trapezoid c.c.'s \cite{xie2012isosceles}, or from Albouy's classification of four-body c.c's with equal masses, \cite{albouy1995symetrie,albouy1996symmetric}, see also section 7.1 in  \cite{cors2012four} for the vortex case), and thus $ r _{ 14 } = r _{ 23 } $.

In the case of the parallelogram opposite sides are equal in length, that is, $ r _{ 12 } = r _{ 34 } $ and $ r _{ 14 } = r _{ 23 } $. %Moreover,  in a parallelogram opposite sides are parallel, this shows that $ \theta _{ 124} = \theta _{ 243 } $ and  $ \theta _{ 142 } = \theta _{ 324 } $.
From equation \eqref{eqn:dziobek} and Lemma \ref{lem:diag_sides} it follows that all the external sides must have equal length. The quadrilateral is then  a rhombus.  
%each diagonal divides the parallelogram into two isosceles triangles, hence $ \theta _{ 142 } = \theta _{ 124 } $, and $ \theta _{ 243 } = \theta _{ 324 } $.
%Since $ m _1 = m _3 $ and  $ m _2 = m _4 $ it follows from equations \eqref{eqn:angle3} and \eqref{eqn:angle4} that $ r _{ 14 } = r _{ 23 } $.

%We will not prove directly  that the  parallelogram is then a rhombus. This however follows from the results in  \cite{perez2007convex}.   

Conversely suppose that $  r _{ 14 } = r _{ 23 } $. Then, from equation \eqref{eqn:angle3} and \eqref{eqn:angle4}
we obtain 
%\begin{equation}
\[   0 = r _{ 14} ^{ - 3 } - r _{ 23 } ^{  - 3 } =  \frac{ \eta _1 } { m _1 m _4 } (-\cot \theta _{ 124 } + \cot 
    \theta_{ 243 }).
\] 
%\end{equation} 
It follows that  $ \theta _{ 124 } = \theta _{ 243 } $ and thus  
$\mathbf{a} $ and $\mathbf{c} $ are parallel. The configuration is either an isosceles trapezoid  or a parallelogram.
In the former case  $ r _{ 14 } = r _{ 23 } $, and $ r _{ 13 } = r _{ 24 } $.  Using equation \eqref{eqn:masses} we find that all the masses are equal, and hence the configuration must be a square. 
In the latter case  $ r _{ 14 } = r _{ 23 } $, and $ r _{12 } = r _{ 34 } $. Reasoning as we did before we find the configuration must be a rhombus. In  either  case  we
deduce that $ r _{ 12 } = r _{ 34 } $.  This completes the proof.

\end{proof}
%-----------------------------
\begin{remark} We conclude this section by noticing that the equations (\ref{eqn:angle1}-\ref{eqn:angle6}) and 
    (\ref{eqn:angle1a}-\ref{eqn:angle6a}) simplify considerably in two important cases.

    If the configuration is co-circular then $ \cot _{ 234 } + \cot _{ 214 } = 0 $ and $ \cot \theta _{ 143 } + \cot \theta _{ 123 } = 0 $, simplifying equations \eqref{eqn:angle5} and \eqref{eqn:angle2a}.

    If the diagonal are perpendicular then $ \theta = \frac{ \pi } { 2 } $, $ \cot \theta = 0 $, and  $\csc \theta = 1 $. In this case the equations (\ref{eqn:angle1}-\ref{eqn:angle6}) and 
    (\ref{eqn:angle1}-\ref{eqn:angle6}) take a simpler form.

It is hoped  that equations (\ref{eqn:angle1}-\ref{eqn:angle6}) and 
    (\ref{eqn:angle1a}-\ref{eqn:angle6a}) for four-body c.c.'s can be helpful in proving further results in the co-circular case, and in the case the diagonals are perpendicular. In the first case it would be interesting, for instance, to  try to reproduce results obtained in \cite{cors2012four}. In the latter case it would be intriguing to try to recover results obtained in \cite{deng2017some} and \cite{corbera2016four}.
\end{remark} 
\section{A remark on co-circular configurations}\label{sec:cocircular}
%%%%%%%%%%%%%%%%%%%%%%%%%%%%%%%%%%%%%%%%%%%%%%%%%%%%%%%%%%%%%%%%%%%%%%%%%
Planarity Condition \ref{cond:planarity5} can also be used to obtain equations for four-body c.c.'s with distances as variables. Although we will not derive such equations here, we will use  Planarity Condition  \ref{cond:planarity5} to explain why  using Ptolemy's condition as a constraint (as it was done by Cors and Roberts in \cite{cors2012four})) is enough to obtain equations for co-circular c.c.'s with distances as variables. Our result is similar to  Lemma 2.1 in \cite{cors2012four}.
Dividing  equation \eqref{eqn:planarity5} by $ (ac + b d +e f) $ yields the condition
\[ 
G=P- Q \cos ^2 \gamma=0 
\]
where 
\[P = ac + b d - e f, \mbox{ and } Q = 4\frac{ abcd}{a c + b d + ef}.\]
The planarity condition is now $ G = 0 $, and central configurations can be viewed as the critical points of $ U $ with the constraints $ I -I _0 = 0 $ and $ G = 0 $.
If $ \lambda $ and $ \sigma $ are Lagrange multipliers this means that we must find the extrema of 
\begin{equation} \label{eqn:G}
    U + \lambda M (I - I _0) + \sigma G 
\end{equation} 
satisfying $ I -I _0 = 0 $ and $ G = 0 $. Let $ \mathbf{r} = (r _{ 12 } , r _{ 13 } , r _{ 14 } , r _{ 23 } , r _{ 24 } , r _{ 34 }) $, and let $ \nabla _{ \mathbf{r} } = \left(\frac{ \partial } { \partial r _{ 12 } }, \ldots , \frac{ \partial } { \partial r _{ 34 } } \right)  $.
Setting the gradient of equation \eqref{eqn:G} equal to zero yields the equations
\begin{equation}\label{eqn:gradG} 
    \nabla _{ \mathbf{r} } U + \lambda M \,\nabla _{ \mathbf{r} } I + \sigma \,\nabla _{ \mathbf{r} } G = 0 
\end{equation} 
satisfying the constraints $ I -I _0 = 0 $ and $ G = 0 $.

Let $r_{ij}$ be one of the mutual distances, then the derivative will have the following form 

\begin{align*} 
    \frac{ \partial G } { \partial r_{ij} } & = \frac{ \partial P } { \partial r_{ij} } -\frac{ \partial Q }{ \partial r_{ij} } \cos ^2 \gamma +2 Q \sin \gamma \cos \gamma \frac{ \partial \gamma } { \partial r_{ij} } \\
    & = \frac{ \partial P } { \partial r_{ij} } +\cos \gamma \left(2 (Q \sin \gamma)  \frac{ \partial \gamma } { \partial r_{ij} } -(\cos \gamma) \frac{ \partial Q }{ \partial r_{ij} }  \right)\\
    & = \frac{ \partial P } { \partial r_{ij} } +\cos \gamma \left( L \sin (\gamma - \phi)   \right) .
\end{align*}
where we defined two functions $ L$ and $ \phi $ such that $ L \cos \phi = 2Q \frac{ \partial \gamma } { \partial r_{ij} } $, $ L \sin \phi = \frac{ \partial Q } { \partial r_{ij} } $, and $ L = \sqrt{ \left(2Q \frac{ \partial \gamma } { \partial r_{ij} } \right)^2 +\left(\frac{ \partial Q } { \partial r_{ij} } \right)^2} $ .
If we restrict the problem to co-circular configurations then $ \cos \gamma = 0 $, and the derivatives of $G$  coincide with those of $P$. We have proven the following lemma.
\begin{lemma}
   Let $ \mathbf{r} $ be any  standard planar convex configurations satisfying $ P (\mathbf{r}) = 0 $, then 
   \[
       \nabla _{ \mathbf{r} } G (\mathbf{r}) = \nabla _{ \mathbf{r} } P (\mathbf{r}).  
   \] 
   In other words, on a configuration for  which both $G $ and $P $ vanish, the
gradients of these two functions are equal.
\end{lemma} 
Consequently,  on co-circular configurations, equation \eqref{eqn:gradG} takes the form 
\[
\nabla _{ \mathbf{r} } U + \lambda M \,\nabla _{ \mathbf{r} } I + \sigma \,\nabla _{ \mathbf{r} } P = 0, 
\]
with the constraints $ I -  I _0 = 0 $, $ G = 0 $, and $ P = 0 $.
This gives another explanation of why   it is possible to use $P$ instead of the Caley-Menger determinant to study co-circular configurations.

\section*{Acknowledgements}
This article had a long gestation before reaching the current form. Donald G. Saari first encouraged me to write c.c.'s equations in terms of angles while I was a Visiting Assistant Professor at UC Irvine for the period 2003-2006. I would like  to thank Alain Albouy  and  Cristina Stoica for helpful comments on an early draft of this work. I am grateful to Shengda Hu for  helpful discussions. This work was supported by an NSERC Discovery grant.  
% Fakesection
%-------------------------------------------
\bibliographystyle{amsplain}
%\AtEveryBibitem{\clearfield{mrnumber}} 
%\bibliographystyle{unsrt}   % this means that the order of references is determined by the order in which the citations appear in the text.
%\nocite{*} 		% The command \nocite{*} causes all items in the database to99229944 be included in the references, regardless of whether or not they are cited in the paper.
\bibliography{references}% list here all the bibliographies that  you need.

\providecommand{\bysame}{\leavevmode\hbox to3em{\hrulefill}\thinspace}
\providecommand{\MR}{\relax\ifhmode\unskip\space\fi MR }
\providecommand{\MRhref}[2]{%
  \href{http://www.ams.org/mathscinet-getitem?mr=#1}{#2}
}
\providecommand{\href}[2]{#2}
\begin{thebibliography}{10}

\bibitem{albouy1995symetrie}
Alain Albouy, \emph{Sym{\'e}trie des configurations centrales de quatre corps},
  Comptes rendus de l'Acad{\'e}mie des sciences. S{\'e}rie 1, Math{\'e}matique
  \textbf{320} (1995), no.~2, 217--220.

\bibitem{albouy1996symmetric}
\bysame, \emph{The symmetric central configurations of four equal masses},
  Contemporary Mathematics \textbf{198} (1996), 131--136.

\bibitem{albouy2006mutual}
\bysame, \emph{Mutual distances in celestial mechanics}, Nelineinaya Dinamika
  [Russian Journal of Nonlinear Dynamics] \textbf{2} (2006), no.~3, 361--386.

\bibitem{albouy2008symmetry}
Alain Albouy, Yanning Fu, and Shanzhong Sun, \emph{Symmetry of planar four-body
  convex central configurations}, Proceedings of the Royal Society of London A:
  Mathematical, Physical and Engineering Sciences \textbf{464} (2008),
  no.~2093, 1355--1365.

\bibitem{behnke1974fundamentals}
Heinrich Behnke, \emph{Fundamentals of mathematics: Geometry}, vol.~2, Mit
  Press, 1974.

\bibitem{corbera2016four}
Montserrat Corbera, Josep~M Cors, and Gareth~E Roberts, \emph{A four-body
  convex central configuration with perpendicular diagonals is necessarily a
  kite}, Qualitative Theory of Dynamical Systems (2016), 1--8.

\bibitem{cors2012four}
Josep~M Cors and Gareth~E Roberts, \emph{Four-body co-circular central
  configurations}, Nonlinearity \textbf{25} (2012), no.~2, 343.

\bibitem{deng2017four}
Yiyang Deng, Bingyu Li, and Shiqing Zhang, \emph{Four-body central
  configurations with adjacent equal masses}, Journal of Geometry and Physics
  \textbf{114} (2017), 329--335.

\bibitem{deng2017some}
\bysame, \emph{Some notes on four-body co-circular central configurations},
  Journal of Mathematical Analysis and Applications \textbf{453} (2017), no.~1,
  398--409.

\bibitem{dziobek1900uber}
Otto Dziobek, \emph{Uber einen merkw{\"u}rdigen fall des
  vielk{\"o}rperproblems}, Astron. Nach \textbf{152} (1900), 33--46.

\bibitem{hampton2014relative}
Marshall Hampton, Gareth~E Roberts, and Manuele Santoprete, \emph{Relative
  equilibria in the four-vortex problem with two pairs of equal vorticities},
  Journal of Nonlinear Science \textbf{24} (2014), no.~1, 39--92.

\bibitem{hobson2004treatise}
E.W. Hobson, \emph{A treatise on plane and advanced trigonometry}, Dover Books
  on Mathematics Series, Dover Publications, 2004.

\bibitem{long2002four}
Yiming Long and Shanzhong Sun, \emph{Four-body central configurations with some
  equal masses}, Archive for rational mechanics and analysis \textbf{162}
  (2002), no.~1, 25--44.

\bibitem{moeckel2001generic}
Richard Moeckel, \emph{{Generic finiteness for Dziobek configurations}},
  Transactions of the American Mathematical Society \textbf{353} (2001),
  no.~11, 4673--4686.

\bibitem{perez2007convex}
Ernesto Perez-Chavela and Manuele Santoprete, \emph{Convex four-body central
  configurations with some equal masses}, Archive for rational mechanics and
  analysis \textbf{185} (2007), no.~3, 481--494.

\bibitem{saari2005collisions}
Donald Saari, \emph{{Collisions, rings, and other Newtonian N-body problems}},
  no. 104, American Mathematical Soc., 2005.

\bibitem{schmidt2002central}
Dieter Schmidt, \emph{Central configurations and relative equilibria for the
  n-body problem}, Classical and Celestial Mechanics (Recife, 1993/1999)
  (2002), 1--33.

\bibitem{xie2012isosceles}
Zhifu Xie, \emph{Isosceles trapezoid central configurations of the newtonian
  four-body problem}, Proceedings of the Royal Society of Edinburgh Section A:
  Mathematics \textbf{142} (2012), no.~3, 665--672.

\end{thebibliography}
%------------------------------------------

\end{document}